\documentclass[english,11pt]{article}

\usepackage[latin9]{inputenc}
\usepackage{color}
\usepackage{amsthm}
\usepackage{amsmath}
\usepackage{amssymb}
\usepackage{bm}
\usepackage{graphicx}
\usepackage{fullpage}

\usepackage[colorlinks,
            linkcolor = blue,
            citecolor = blue,
            urlcolor = blue]{hyperref}

\makeatletter

\theoremstyle{plain}
\newtheorem{thm}{\protect\theoremname}[section]
\newtheorem{cor}[thm]{\protect\corollaryname}
\newtheorem{lem}[thm]{\protect\lemmaname}

\theoremstyle{definition}
\newtheorem{defn}[thm]{\protect\definitionname}

\def\01{\{0,1\}}

\newcommand{\floor}[1]{\lfloor{#1}\rfloor}

\newcommand{\ket}[1]{|#1\rangle}
\newcommand{\bra}[1]{\langle#1|}

\newcommand{\set}[1]{\{#1\}}

\newcommand{\Tr}{\mbox{\rm Tr}}

\newcommand{\event}{\mathcal{E}}

\newcommand{\ns}{\mathrm{ns}}

\makeatother

\usepackage{babel}
  \providecommand{\definitionname}{Definition}
  \providecommand{\lemmaname}{Lemma}
  \providecommand{\propositionname}{Proposition}
  \providecommand{\theoremname}{Theorem}
  \providecommand{\corollaryname}{Corollary}
  \providecommand{\conjecturename}{Conjecture}
  \providecommand{\observationname}{Observation}
  
\newcommand{\qindep}{\alpha^{*}} 
\newcommand{\qchrom}{\chi^{*}} 

\begin{document}

\title{Multi-party zero-error classical channel coding with entanglement}

\author{Teresa Piovesan\thanks{Centrum Wiskunde \& Informatica (CWI),  Amsterdam, The Netherlands.}, Giannicola Scarpa\thanks{Universitat Autonoma de Barcelona, Spain. }, \\ and Christian Schaffner\thanks{Institute for Logic, Language and Computation (ILLC), University of Amsterdam, The Netherlands and Centrum Wiskunde \& Informatica (CWI),  Amsterdam, The Netherlands. }}

\maketitle

\begin{abstract}
We study the effects of quantum entanglement on the performance of two classical zero-error communication
tasks among multiple parties. 
Both tasks are generalizations of the two-party zero-error channel-coding problem, where a sender and a receiver want to perfectly communicate messages through a one-way classical noisy channel.
If the two parties are allowed to share entanglement, there are several positive results that show the existence of channels for which they can communicate strictly more than what they could do with classical resources.

In the first task, one sender wants to communicate a common message to multiple receivers.
We show that if the number of receivers is greater than a certain threshold then entanglement does not allow for an improvement in the communication for any finite number of uses of the channel.
On the other hand, when the number of receivers is fixed, we exhibit a class of channels for which entanglement gives an advantage.

The second problem we consider features multiple collaborating senders and one receiver. 
Classically, cooperation among the senders might allow them to communicate on average more messages than the sum of their individual possibilities. 
We show that whenever a channel allows single-sender entanglement-assisted advantage, then the gain extends also to the multi-sender case. 
Furthermore, we show that entanglement allows for a peculiar  amplification of information which cannot happen classically,  
for a fixed number of uses of a channel with multiple senders.
\end{abstract}

\section{Introduction}

Suppose Alice wants to send a message to Bob but they can communicate only through a one-way classical noisy channel. How much information can she send to him on average, such that Bob learns Alice's message with zero probability of error? 
This is known as the zero-error channel-coding problem.
Since its introduction by Shannon \cite{Shannon:1956}, several generalizations to multi-party settings have been proposed and a large research area in the intersection of information theory and combinatorics has been developed (see \cite{Korner} for a survey). 

Recently, Cubitt et al.~\cite{CLMW09} introduced the entanglement-assisted version of the two-party problem: How much data can Alice send to Bob with zero error when they are connected through a one-way classical noisy channel and share an entangled quantum state?
Does entanglement allow to transmit more information? 
An affirmative answer for the latter question was given in \cite{LMMOR} and \cite{Briet:2012} where are shown channels exhibiting a separation between the entanglement-assisted and classical communication.
Two parties can communicate strictly more if they share a certain entangled state than what they could do in the classical case.
Note that this separation holds only for the special case where we want the communication to succeed with zero error.
As shown in \cite[Theorem 1]{Bennett:2002}, sharing entanglement does not provide any advantage in the more general case of vanishing error probability, where we ask the probability of error to asymptotically go to zero as the number of uses of the channel goes to infinity.

Inspired by these results, we continue this line of research by investigating the effect of quantum entanglement in two multi-party zero-error communication scenarios.

In the first situation we study, one sender is connected through identical classical channels to multiple receivers. This is known as a \emph{compound channel}.
We prove that, for any fixed number of uses of the channel, entanglement might improve the communication only up to a certain number of receivers  (Theorem~\ref{thm:qmonogamy}). The effect is due to the monogamy of entanglement (or more generally, of non-signaling correlations~\cite{MAG06}). Indeed, if the number of receivers is greater than some threshold (which depends uniquely on the channel and the number of channel uses), we show that entanglement does not help. However, for any constant number of receivers, we can build a compound channel (based on a family studied in \cite{Briet:2012}) for which there is a separation between the entanglement-assisted and classical setting (Corollary~\ref{cor:compound}).

In the second problem, multiple senders cooperate to communicate with a single receiver through identical classical channels.
Our Theorem~\ref{thm:sep_c_l} shows that there exist channels for which entanglement increases the amount of communication that can be sent, independently from the number of senders.
More surprisingly, there are channels for which entanglement allows a joint strategy among the senders that is strictly better than the sum of the individual strategies (Theorem~\ref{thm:disjointalpha}). 
If each sender is able to transmit $m$ messages using entanglement, there is a joint entanglement-assisted strategy for $\ell$ senders that allows to communicate strictly more than $\ell \cdot m$ messages. 
This is in contrast with the classical case where this phenomenon cannon happen.

The rest of the paper is organized as follows. 
In Section \ref{sec:prel} we introduce the basic notation and the two-party problem.
In Section \ref{sec:rec} we study the entanglement-assisted one-sender and multi-receiver situation. 
In Section \ref{sec:send} we present the effect of entanglement when there are multiple cooperating senders and one receiver.
Section~\ref{sec:con} contains the conclusions and some open questions.

\section{Preliminaries}\label{sec:prel}

\subsection{Notation and basics of graph theory}\label{sec:not}

We denote with $[n]$ the set $\{1,\dots,n\}$, with $\Pi(n)$ the symmetric group over $[n]$, with $\delta_{i,j}$ the Kronecker delta function ($\delta_{i,j}=0$ if $i\neq j$ and $\delta_{i,j}=1$ if $i=j$) and with $I$  the identity matrix. 
Consider a positive semidefinite operator $\rho$ that acts on a bipartite finite-dimensional Hilbert space $\mathcal{H}_A \otimes \mathcal{H}_B$. We denote with $\Tr_A(\rho)$ the partial trace of $\rho$ over the subspace $\mathcal{H}_A$.
Moreover suppose $\rho$ acts on a finite-dimensional $\ell$-partite Hilbert space $\mathcal{H}^{\otimes \ell}$
which we denote as $B_1 \otimes B_2 \otimes \dots \otimes B_\ell$. With $\Tr_{B_{-k}}(\rho)$ we denote the partial trace of $\rho$ over all the subspaces but the $k$-th one, \emph{i.e.}, $\Tr_{B_{-k}}(\rho) = \Tr_{B_1, \dots, B_{k-1}, B_{k+1}, \dots, B_\ell}(\rho)$.
Throughout the paper, the logarithms are binary and the graphs are assumed to be undirected and simple.

For any graph $G$, we denote with $V(G)$ and $E(G)$ its vertex and edge set. If two vertices $u,v \in V(G)$ are equal or adjacent we write $u \simeq v$, analogously we denote $u \sim v$ if they are distinct and adjacent.
The complement graph of $G$, denoted by $\overline G$, is the graph on the same vertex set as $G$ where two distinct vertices are adjacent if and only if they are non-adjacent in $G$.
An \emph{independent set} of $G$ is a subset of $V(G)$ that contains only pairwise non-adjacent vertices. The {\em independence number} $\alpha(G)$ is the maximum cardinality of an independent set of $G$.
A coloring is a partition of the vertex set into independent sets. The {\em chromatic number} $\chi(G)$ is the minimum cardinality of a coloring. 
A \emph{clique} is a set of pairwise adjacent vertices. 
The {\em edge-clique cover number} $\theta_e(G)$ is the smallest number of cliques that together cover all the edges of the graph $G$. We denote with $\theta'_e(G)$ the edge-clique cover number of $G$ plus the number of isolated vertices of $G$.
An \emph{orthogonal representation} of a graph $G$ is a map from the vertex set into the $d$-dimensional unit sphere such that adjacent vertices are mapped to orthogonal vectors. The minimum dimension $d$ for which such a representation exists is denoted as the \emph{orthogonal rank} $\xi(G)$.

The graph $K_t$ is the \emph{complete graph} on $t$ vertices, where every pair of distinct vertices is adjacent.
The {\em orthogonality graph} $\Omega_k$ has as vertex set all the vectors in $\{\pm 1\}^k$ and two vertices are adjacent if the corresponding vectors are orthogonal.

The {\em strong product graph} of $G$ and $H$ is denoted by $G \boxtimes H$. It has vertex set $V(G) \times V(H)$ (where $\times$ denotes the Cartesian product) and a pair of distinct vertices $ux, vy$ is adjacent if $u \simeq v$ in $G$ and $x \simeq y$ in $H$. 
The {\em $n$-th strong graph power} of $G$, denoted by $G^{\boxtimes n}$, is the strong product graph of $n$ copies of $G$. Its vertex set is the Cartesian product of $n$ copies of $V(G)$ and the pair of distinct vertices ${(u_1\dots, u_n),(v_1,\dots,v_n)}$ forms an edge in $G^{\boxtimes n}$ if $u_i \simeq v_i$ in $G$ for all $i \in [n]$.
We denote with $G^{+t}$ the disjoint union of $t$ copies of $G$, \emph{i.e.}, $V(G^{+t}) = V(G) \times [t]$ and the pair $ui, vj$ is adjacent if $u \sim v$ in $G$ and $i = j \in [t]$.
The {\em Cartesian product graph} of $G$ with a complete graph $K_t$, denoted by $G \square K_t$, has vertex set $V(G) \times [t]$ and the pair $ui,vj$ is adjacent if $u\sim v$ in $G$, $i=j\in [t]$ or if $u=v \in V(G)$, $i \neq j \in [t]$. 

The Lov\'asz theta number~\cite{Lovasz} of a graph $G$ is equal to 
$$\vartheta(G) = \max \sum_{u,v \in V(G)}X_{uv} \; \text{ s.t.} \sum_{u\in V(G)}X_{uu}=1, \; X_{uv} = 0 \quad \forall uv \in E(G),\; X \succeq 0;$$
where $X \succeq 0$ means that $X$ is a positive semidefinite matrix.
As $\vartheta(G)$ is the optimal value of a positive semidefinite program, it can be computed up to any approximation in polynomial time in the number of vertices $|V(G)|$.
For any graph $G$ we have $\alpha(G) \leq \vartheta(G) \leq \chi(\overline G)$.
Among the many useful properties of Lov\'asz theta number, we will use that $\vartheta(G \boxtimes H) = \vartheta(G) \vartheta(H)$ and $
\vartheta (\overline{G \boxtimes H}) = \vartheta (\overline G) \vartheta(\overline{H})$  for every pair of graphs $G$ and $H$. Moreover $\vartheta$ is monotone non-decreasing under taking subgraphs and $\vartheta(K_t) = 1$, $\vartheta(\overline{K}_t) = t$ for all $t \in \mathbb{N}$ (see \cite{Knuth:1993} for a survey on the properties of $\vartheta$).

\subsection{The zero-error channel capacity}

Imagine the following scenario: Alice wants to communicate a message to Bob but they can only use a one-way classical noisy channel $\mathcal{N}$.
The channel is fully characterized by its finite input set $V$, its finite output set $W$ and a probability distribution $\mathcal N (\cdot|v)$ for every $v \in V$, where $\mathcal N (w | v)$ is the probability of Bob receiving outcome $w \in W$ given that Alice has sent input $v \in V$. 

To communicate a message $i \in [m]$, Alice uses an encoding function $C: [m] \to V$ and sends $C(i)$ through the channel. Bob receives output $w$ with probability $\mathcal N (w|C(i))$ and uses a decoding function $D: W \to [m]$ to get a message $D(w)$. 
We consider only the zero-error scenario, where $D(w)$ must always be equal to Alice's original message $i$.
What is the maximum size $m$ of a message set that Alice can use?
As shown by Shannon~\cite{Shannon:1956}, this problem can be reformulated in graph-theoretic terms.
  
To a channel $\mathcal N$ with input set $V$ and output set $W$, we associate a {\em confusability graph} $G$, where $V(G)=V$ and the pair $u,v \in V$ forms an edge if there exists a $w\in W$ such that $\mathcal N (w|u) \mathcal N (w|v) > 0$.
In other words, $u$ and $v$ form an edge if  Bob can potentially confuse these inputs and every output of the channel can be naturally associated with a clique of $G$, a set of pairwise confusable inputs.
Since we want the communication between Alice and Bob to be zero-error, the largest size $m$ of a message set they can employ with one use of the channel is the largest set of pairwise non-confusable inputs, \emph{i.e.}, $m$ is the independence number $\alpha(G)$.
(Equivalently, Alice can communicate at most $\log \alpha(G)$ bits of information.) 
The confusability graph of $n$ channel uses is represented by the strong graph power $G^{\boxtimes n}$.
Shannon~\cite{Shannon:1956} showed that using the channel multiple times can be on average more efficient than using it only once. The {\em Shannon capacity} of a (channel with confusability) graph $G$, 
$$c(G) = \lim_{n\to\infty}\frac{1}{n}\log\alpha(G^{\boxtimes n}),$$
is the maximum average number of bits that can be communicated with zero-error.
By super-multiplicativity of $\alpha(\cdot)$ and Fekete's Lemma,\footnote{ 
Super-multiplicativity says that $\alpha(G^{\boxtimes m+m'}) \geq \alpha(G^{\boxtimes m}) \alpha(G^{\boxtimes m'})$ for all $m,m' \in \mathbb{N}$.
Fekete's Lemma says that if a sequence $(a_m)_{m\in \mathbb{N}}$ is super-additive, $a_{m+m'} \geq a_m + a_{m'}$ $\forall m,m' \in \mathbb{N}$, then the limit of the sequence $(a_m/m)_{m\in \mathbb{N}}$ exists and $\lim_{m \to \infty} a_m/m = \sup_{m} a_m/m$. 
} the Shannon capacity is well-defined and it is also equal to $c(G) = \sup_{n}\frac{1}{n}\log\alpha(G^{\boxtimes n})$. 
Computing the Shannon capacity is a notoriously hard problem and even the value of the Shannon capacity of the 5-cycle was unknown until Lov\'asz \cite{Lovasz} calculated it using the parameter $\vartheta(G)$, commonly referred to as the Lov\'asz theta number, which has the property that $c(G) \le \log \vartheta(G)$.
The computational complexity of $c(G)$ is unknown (see for example \cite{Alon:2006}).

Suppose that Alice and Bob can share an entangled state and use it as additional resource in their communication protocol.
For simplicity let us consider a single use of the channel.
To send a message $i \in [m]$ to Bob, Alice performs a measurement $ \{ A_i^u\}_{u\in V} $ (which depends on~$i$) on her part of the entangled state, and sends its outcome $u\in V$ through the channel $\mathcal{N}$.
With probability $\mathcal N(w|u)$, Bob receives $w \in W$ and uses this information to perform a measurement $ \{ B_{w}^{\hat{i}}\}_{{\hat{i}} \in [m]}$ on his side of the entangled state getting the message $\hat{i}$ as outcome. Bob outputs $\hat{i}$ and in the zero-error scenario we require $\hat{i}$ to be equal to $i$ with zero probability of error.
In other words, we want an entangled state $\ket{\psi}$ in $\mathcal{H}_A \otimes \mathcal{H}_B$ and a measurement with POVM elements $\{ A_i^u \}_{u\in V}$ for every $i\in [m]$ such that 
$$ \Tr_A ( A_i^u \otimes I_B \ket{\psi} \bra{\psi}) \ \perp \  \Tr_A ( A_j^v \otimes I_B \ket{\psi} \bra{\psi}) \text{ for all } i \neq j\in [m], u \simeq v \in V(G),
$$
where $G$ is the confusability graph of $\mathcal N$.
Bob has to be able to perfectly distinguish between distinct messages $i$ and $j$, whenever Alice channel's input might be confusable.
This entanglement-assisted variant has been introduced by Cubitt et al.~\cite{CLMW09}. Setting $\rho_i^u := \Tr_A ( A_i^u \otimes I_B \ket{\psi} \bra{\psi})$ they obtained the following concise definition of the single channel use and of the entangled Shannon capacity.

\begin{defn}[\cite{CLMW09}] \label{def:qindep}
For a graph $G$, the \emph{entanglement-assisted independence number} $\alpha^*(G)$ is the maximum $m\in \mathbb{N}$ such that there exist positive semidefinite operators $\{\rho_i^u : i\in[m],\, u\in V(G)\}$ and~$\rho$ acting on a finite-dimensional Hilbert space $\mathcal{H}$ such that
\begin{align*}
\Tr(\rho) &= 1 \, ,\\
\sum_{u\in V(G)} \rho^u_i &= \rho \quad \forall i\in[m] \, ,\\
\rho^u_i\rho^v_j &= 0 \quad \forall i\ne j \in [m],\, \forall u \simeq v \in V(G) \, .
\end{align*}

The \emph{entangled Shannon capacity} is $c^*(G) = \lim_{n\to\infty}\frac{1}{n}\log\alpha^*(G^{\boxtimes n})$ and indicates the maximum average number of bits that Alice can perfectly communicate to Bob using entanglement.
\end{defn}

Notice that we have only shown that from a protocol we can find matrices that satisfy Definition~\ref{def:qindep}. For the proof of the reverse statement, and thus of the correctness of Definition \ref{def:qindep}, we refer to \cite{CLMW09} and for more details to \cite{Briet:2013}.
Since $\alpha^*$ is super-multiplicative~\cite{CLMW09}, $c^*(G)$ is well-defined by Fekete's lemma and can be alternatively expressed as $c^*(G) = \sup_{n}\frac{1}{n}\log\alpha^*(G^{\boxtimes n})$.

The Lov\'asz theta number $\vartheta(G)$ 
is also an upper bound for the entangled Shannon capacity and hence $c(G) \le c^*(G) \leq \log \vartheta(G)$ holds.
In fact~\cite{ZEC-theta,Duan:2013} proved that $\alpha^*(G) \leq \floor{\vartheta(G)}$ and using the multiplicativity of $\vartheta(G)$ under strong graph products concluded that $c^*(G) \leq \log \vartheta(G)$.
This means that the Lov\'asz theta number cannot be used to separate the classical Shannon capacity $c(G)$ from the entangled capacity $c^*(G)$ and therefore to find channels for which entanglement improves the communication.
Despite this fact, there exist examples of graphs for which $\alpha(G) < \alpha^*(G)$ (see~\cite{CLMW09}) and $c(G) < c^*(G)$ (see~\cite{LMMOR} and~\cite{Briet:2012}).

\section{Multiple receivers}\label{sec:rec}
In this section, we introduce a first generalization of the channel-coding problem.
Suppose that there are $\ell$ receivers that want to decode a common message sent by a single sender, as for example in TV broadcasting.
This is known as the \emph{compound channel} model.
We focus on the zero-error case where each receiver can decode the message with zero probability of error. This setting was studied in \cite{Cohen:1990} as a generalization of the zero-error Shannon capacity (see \cite{Sinaimeri} for a detailed description).  
After a brief introduction of the classical scenario we analyze the entanglement-assisted variant of this problem.

Consider a family of channels $\bm{\mathcal{N}} = \{\mathcal{N}_{1},\dots,\mathcal{N}_{\ell} \}$ with the same input set $V$ where $\mathcal N_k$ connects the sender with the $k$-th receiver.
A common input $v \in V$ is sent to all the receivers and the $k$-th receiver gets the output $w_k$ according to the distribution $\mathcal N_k(w_k|v)$.
The goal is for each receiver to retrieve the original input $v$ with zero probability of error. 
As for the single-channel case, this problem can be treated from a graph-theoretical perspective 
associating to each channel $\mathcal N_k$ a confusability graph $G_k = (V,E_k)$.
Note that the family of graphs $\mathcal{G} = \{G_1,\dots G_\ell \}$ share the same vertex set, since the input set is common.
If we find a subset of the input set $V$ which is an independent set in every graph  $G_k \in \mathcal G$, this subset forms a set of non-confusable inputs for each of the receivers and thus can be used for zero-error communication. 
In general, 
define $\alpha(\mathcal G,n)$ to be the maximum cardinality of a set in $V^n$ which is an independent set in every graph $G_k^{\boxtimes n}$, with $G_k \in \mathcal G$.  
Hence, $\alpha(\mathcal G,n)$ is the maximum number of messages that can be transmitted perfectly to each of the receivers with $n$ channel uses.
Observing that $\alpha(\mathcal G,\cdot)$ is super-multiplicative (\emph{i.e.}, $\forall n_1,n_2 \in \mathbb{N}$: $\alpha(\mathcal G,n_1+n_2) \geq \alpha(\mathcal G,n_1) \alpha(\mathcal G,n_2)$), the Shannon capacity of a family of graphs $\mathcal G$ is well-defined as $c(\mathcal G) = \lim_{n \to \infty} \frac{1}{n} \log \alpha(\mathcal G,n)$. 
Computing this quantity appears to be hard, since it comprises the computation of the Shannon capacity of a single graph as special case.
However, some positive results have been obtained for a slightly different task. Suppose that computing the capacity of every element of the family $\mathcal G$ is easy. Is it possible to determine the capacity of the whole family of graphs? An affirmative answer is given in \cite{Gargano:1994}.

We study the entanglement-assisted version of this problem, 
 where the sender shares a single entangled state with all the receivers,
focusing on the particular instance where all the channels are equal.
Thus, in the rest of the section we assume that all the receivers are connected to the sender through the same channel $\mathcal N$ with confusability graph $G$. This allows us to work directly on $G$, introducing the notation $\alpha_{1,\ell}(G) := \alpha(\mathcal G,1)$ and $c_{1,\ell}(G) := c(\mathcal G)$. 
Note that in the classical case this leads to the trivial situation $\alpha_{1,\ell}(G^{\boxtimes n}) = \alpha(G^{\boxtimes n})$ for each $n \in \mathbb{N}$ (any independent set of $G$ is an independent set for the whole family) and thus $c_{1,\ell}(G) = c(G)$.

\begin{figure}[t]
\begin{center}
 \includegraphics[width=8cm]{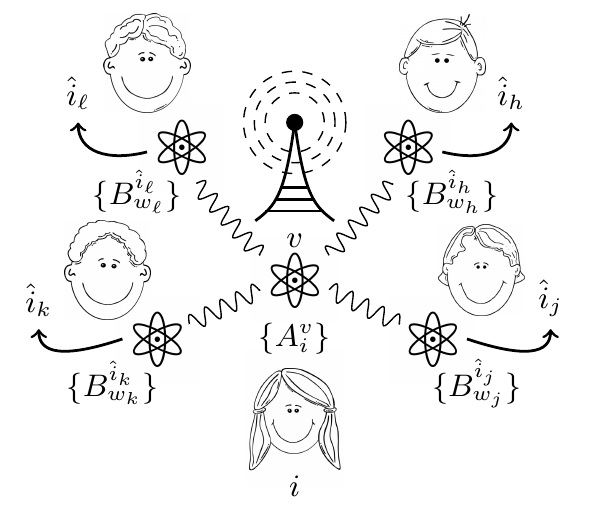}
 \caption{The figure describes an entanglement-assisted compound channel instance, with a single use of the channels. Alice wants to send the same message $i \in [m]$ to every Bob, with whom she shares an entangled state. She performs a measurement $\{A_{i}^{v}\}$ on her part of the entangled state obtaining ${v}$ as outcome. This outcome is used as input for all the channels and the $k$-th Bob receives $w_k$ from $\mathcal N_k$. Each Bob performs a measurement $\{B_{{w}_k}^{{\hat{i}}_k}\}$ (which depends on ${w}_k$) on his part of the entangled state, getting a message ${\hat{i}}_k$ as outcome.
The protocol works if ${\hat{i}}_k$ is equal to ${i}$ for every $k$, \emph{i.e.}, every Bob is able to perfectly learn Alice's original message.
}
\label{fig:broadcast}
\end{center}
\end{figure}

The protocol for the entanglement-assisted compound channel is described in Figure~\ref{fig:broadcast}. 
As in the single-receiver case, the protocol depends only on the confusability graph of the channel and we can define the following quantities.


\begin{defn}\label{def:qcomp}
For a graph $G$, the \emph{entanglement-assisted compound independence number} with $\ell$ receivers $\alpha_{1,\ell}^*(G)$ is defined as the maximum $m\in \mathbb{N}$ such that there exist positive semidefinite operators $\{\rho_i^u, i\in[m],\, u\in V(G)\}$ and $\rho$ acting on a finite-dimensional Hilbert space $\mathcal{H}^{\otimes \ell}$  such that
\begin{align*}
\Tr(\rho) &= 1 \, ,\\
\sum_{u\in V(G)} \rho^u_i &= \rho \quad \forall i\in[m] \, , \\
\Tr_{B_{-k}}(\rho^u_i) \; \Tr_{B_{-k}}(\rho^v_j) &= 0 \quad \forall k\in [\ell], \forall i\ne j,\, \forall u \simeq v \in V(G) \, .
\end{align*}
The \emph{entanglement-assisted compound Shannon capacity} with one sender and $\ell$ receivers is $$c_{1,\ell}^*(G) = \lim_{n\to\infty}\frac{1}{n}\log\alpha_{1,\ell}^*(G^{\boxtimes n}).$$
\end{defn}

\noindent
Operationally, $\alpha_{1,\ell}^*(G)$ is the maximum number of messages that Alice can perfectly communicate to all the Bobs through identical channels with confusability graph $G$ and an entangled state.
This interpretation allows us to immediately see that $\alpha_{1,\ell}^*(G)$ is super-multiplicative with respect to strong graph product, that is $\alpha_{1,\ell}^*(G \boxtimes H) \ge \alpha_{1,\ell}^*(G)\alpha_{1,\ell}^*(H)$.
Indeed, suppose that two channels $\mathcal N_1$ and $\mathcal N_2$ allow an entanglement-assisted compound channel protocol with $\ell$ receivers that communicates $m_1$ and $m_2$ messages, respectively.
Then for the combined channel $\mathcal N_1 \otimes \mathcal N_2$, with confusability graph $G \boxtimes H$, the sender can communicate $m_1 \times m_2$ messages by running subsequently the two protocols.
This observation implies by Fekete's lemma that $c_{1,\ell}^*(G)$ is a well-defined quantity.

\subsection{Entanglement for finite number of channel uses}\label{sec:rec-neg}


In this section we investigate a setting where entanglement does not provide any  advantage.
We prove that for any fixed channel and number of channel uses, there exists a number of receivers for which entanglement does not improve over classical communication.
Recall that $\theta'_e(G)$ denotes the edge-clique cover number of $G$ plus the number of isolated vertices of $G$.

\begin{thm} \label{thm:qmonogamy}
For any graph $G$, if $\ell \geq \theta'_e(G)$ then $\alpha^*_{1,\ell}(G) = \alpha(G).$
\end{thm}
 
\noindent
This theorem follows directly from Theorem~\ref{thm:monogamy} (in Section~\ref{nonsignaling}), where we prove an analogous result for the situation when the players can use arbitrary non-signaling correlations.
In words, a probability distribution is non-signaling if the marginal distribution of the output of each subset of parties depends only on the corresponding inputs. 
Measuring an entangled quantum state results in a specific class of non-signaling correlations.


%
%
%

\subsubsection{Non-signaling capacity}\label{nonsignaling}

Here we define 
 a \emph{non-signaling} version of the compound
independence number, $\alpha^{\ns}_{1,\ell}(G)$ (Definition \ref{def:nscomp}), and use it to prove Theorem~\ref{thm:qmonogamy}.
%
%
First recall the mathematical definition of non-signaling probability distributions.


\begin{defn}
A $n$-partite probability distribution $P(a_1,a_2,\dots,a_n|x_1,x_2,\dots,x_n)$ is called {\em non-signaling} if for all outputs $a_1,a_2,\dots,a_n$ and all inputs $x_1,x_2,\dots,x_n$, the marginal distribution for each subset of parties $I = \{i_1,i_2 \dots, i_k \} \subseteq [n]$ only depends on the corresponding inputs
$$P(a_{i_1}, a_{i_2}, \dots, a_{i_k} | x_1,x_2,\dots,x_n) = P(a_{i_1},a_{i_2}, \dots, a_{i_k} |x_{i_1},x_{i_2}, \dots, x_{i_k}).$$
\end{defn}


Since any entanglement-assisted strategy is also non-signaling, the parties can always communicate at least as much information using a non-signaling strategy as they can using entanglement.

Since we are studying the problem of sending information over a channel with zero-error, we have restricted our attention to the properties of the confusability graph of the channel.
However, many channels can have the same confusability graph and, unlike the classical and entanglement-assisted capacities, the non-signaling capacity depends on the particular channel.
For our purposes we are interested in the particular channel that minimizes the number of outputs while keeping the same confusability graph. Notice that every output of a channel defines a clique or an isolated vertex in the confusability graph. 
Therefore, we fix a edge-clique covering of the confusability graph of minimum cardinality $\theta_e(G)$ (which might not be unique), we add the isolated vertices to obtain a clique covering of cardinality $\theta'_e(G)$, and we consider the channel that has $\theta'_e(G)$ outputs. In other words, we take a channel which has one output per element of the edge-clique covering plus one output per isolated vertex. 


For the two-party zero-error channel-coding problem, the non-signaling version was studied in \cite{CLMW09} (see also \cite{CLMW11}). 
There it is shown that the non-signaling zero-error channel capacity has an elegant closed formula that can be computed from the description of the channel.
In the two-party case, a protocol for a single use of the channel $\mathcal N$ (with confusability graph $G$) is the following.
%
%
Alice wants to transmit the
message $i \in [m]$. As a resource, Alice and Bob have access to a
non-signaling distribution $P(v, \hat{i} | i, c)$ where Alice inputs $i$
and obtains a vertex $v \in V(G)$ that she sends over the channel. 
Bob's channel output can be interpreted as a clique (or isolated vertex) $c \subseteq V(G)$ which yields final output $\hat{i} \in [m]$. 
The transmission is successful if we always have that $i=\hat{i}$.

We can now directly extend this protocol to the compound channel with $\ell$ receivers, where Alice and the Bobs share a $\ell +1$-partite non-signaling probability distribution. 
To communicate message $i \in [m]$, Alice inputs $i$ to the non-signaling distribution $P(v, \hat{i}_1, \ldots, \hat{i}_\ell | i, c_1, \ldots, c_\ell)$ and uses her output $v \in V(G)$ as input of the channel. 
For each Bob, his channel output $c_k \subseteq V(G)$, where $v \in c_k$, is either a clique or an isolated vertex of $G$ and if it is used as input to the non-signaling distribution it gives $\hat{i}_k$ as output. 
The protocol works if every single Bob learns $i$ with zero probability of error, that is $i = \hat{i}_k$ for every $k \in \mathbb{N}$.

\begin{defn}\label{def:nscomp}
The \emph{non-signaling compound independence number $\alpha^{\ns}_{1,\ell}(G)$} is the maximum $m \in \mathbb{N}$ such that there exists a non-signaling
distribution 
\[P(v, \hat{i}_1, \ldots, \hat{i}_\ell | i, c_1, \ldots, c_\ell)\] between Alice and Bob$_1$, \ldots, Bob$_\ell$, where $i \in [m]$ and $c_k \subseteq V(G)$ are elements of a fixed clique covering of cardinality $\theta'_e(G)$.
Additionally, we require for all $i \in [m]$ the following: If vertex $v$ is contained in $c_k$ for all $k \in [\ell]$ and there exists a $k' \in [\ell]$ such that $i \neq \hat{i}_{k'}$, then $P(v, \hat{i}_1, \ldots, \hat{i}_\ell | i, c_1, \ldots, c_\ell)=0$. 
\end{defn}

The last requirement imposes the perfect correctness of the protocol. 
Every Bob must output the correct message $i$ in case he 
received as channel output a $c_k$ which is compatible with Alice's channel input $v$ (\emph{i.e.},~if $v \in c_k$ for every $k \in [\ell]$).

As already mentioned, since every quantum strategy is also non-signaling, for every graph $G$ and $\ell \in \mathbb{N}$, it holds that $\alpha^{\ns}_{1,\ell}(G) \geq \alpha^*_{1,\ell}(G)$.

\begin{thm} \label{thm:monogamy}
For all graphs $G$, if $\ell \geq \theta'_e(G)$ then
$$\alpha^{\ns}_{1,\ell}(G) = \alpha^*_{1,\ell}(G) = \alpha(G).$$
\end{thm}

In order to prove the theorem, we use the monogamy of non-signaling
distributions as derived in~\cite{MAG06}. For convenience, we
reproduce the definition and result here.

\begin{defn} \cite{MAG06} \label{def:shareable}
A non-signaling probability distribution $P(a,b | x,y)$ is called
\emph{$\ell$-shareable with respect to Bob}, 
if there exists an $(\ell+1)$-partite non-signaling probability distribution\\
$Q(a,b_1,\ldots,b_\ell|x,y_1,\ldots, y_\ell)$ such that:
\begin{enumerate}
\item For all permutations $\pi \in \Pi(\ell)$, we have that 
\[ Q(a,b_{\pi(1)},\ldots,b_{\pi(\ell)}|x,y_{\pi(1)},\ldots,y_{\pi(\ell)}) =
Q(a,b_1,\ldots,b_\ell|x,y_1,\ldots,y_\ell) \, .
\]
\item It holds that
\[\sum_{b_2,\ldots,  b_\ell} Q(a,b_1,\ldots,b_\ell |
  x,y_1,\ldots,y_\ell) = P(a,b_1 | x,y_1) \, . \]
\end{enumerate}
\end{defn}
Note that if both conditions hold, we have that for all $k \in [\ell]$
\begin{align} \label{eq:shareable} 
\sum_{b_1,\ldots,
    b_{k-1},b_{k+1}, \ldots, b_\ell} \!\!\!\!\! Q(a,b_1,\ldots,b_\ell |
  x,y_1,\ldots,y_\ell) = P(a,b_k | x,y_k)  \, .
\end{align}

\begin{thm} \label{thm:ns} \cite{MAG06}
Let $Y$ be the set of different values for the input $y$ and suppose $\ell \geq |Y|$. If $P(a, b|x, y)$ is a non-signaling
  distribution which is $\ell$-shareable with respect to Bob, then
  $P(a,b|x,y)$ admits a local hidden variable model. Formally, there exists a distribution $Q(\lambda)$ over the hidden variables $\lambda$ as well as local strategies $A(a|x,\lambda)$ for Alice and $B(b|b,\lambda)$ for Bob such that $P(a,b|x,y) = \sum_{\lambda} Q(\lambda) A(a|x,\lambda) B(b|y, \lambda)$.
\end{thm}

\begin{proof}
Assume without loss of generality that $Y=\set{1,2,\ldots,|Y|}$.
The idea of the proof is to ask all possible questions $y=1,2,\ldots,|Y|$ to $|Y|$ different Bobs (which is possible because $\ell \geq |Y|$) and use their answers $b_1,\ldots,b_{|Y|}$ to these questions as hidden variables $\lambda$.

Assume for now that $\ell = |Y|$. 
Let us fix the questions to the $\ell$ Bobs as $y_1=1, y_2=2, \ldots, y_\ell = \ell$ and abbreviate this event with $\event$.
We can then write
\begin{align*}
P(a,b | x,y) 
&\stackrel{\eqref{eq:shareable}}{=}  \sum_{\stackrel{b_1,\ldots,b_\ell}{b_y = b}} Q(a, b_1, \ldots, b_\ell | x, \event) \\
&= \sum_{b_1,\ldots,b_\ell} 
Q(b_1,\ldots,b_\ell|x,\event)
\cdot Q(a | b_1, \ldots, b_\ell, x,\event) \cdot \delta_{b,b_y} \, .
\end{align*}
Due to non-signaling, $Q(b_1,\ldots,b_\ell|x,\event) = Q(b_1,\ldots,b_\ell|\event)=Q(\lambda|\event)$. The conditional distribution $Q(a | b_1, \ldots, b_\ell, x,\event)$ defines Alice's strategy $A(a|\lambda, x, \event)$. Bob's strategy $B(b|\lambda, y, \event)$ is defined by giving the answer $b=b_y$ of the $y$-th Bob. In summary, we obtain a local-hidden-variable representation of $P$:
\begin{align*}
P(a,b | x,y) 
&= \sum_{\lambda} 
Q(\lambda|\event) \cdot A(a | \lambda, x, \event) \cdot B(b | \lambda, y, \event)  \, .
\end{align*}

In case that $\ell > |Y|$, we observe that $\ell$-shareability of $P(a,b|x,y)$ implies $|Y|$-shareability.
Hence, the above proof applies.
\end{proof}

\begin{proof} [Proof of Theorem~\ref{thm:monogamy}]
Let $P(v, \hat{i}_1, \ldots, \hat{i}_\ell | i, c_1, \ldots, c_\ell)$ be the optimal
non-signaling probability distribution achieving
$\alpha^{\ns}_{1,\ell(G)}$. We define the following distribution
\[ Q(v, \hat{i}_1, \ldots, \hat{i}_\ell | i, c_1, \ldots, c_\ell) := \sum_{\pi \in
  \Pi(\ell)} \frac{1}{|\Pi(\ell)|} P(v, \hat{i}_{\pi(1)}, \ldots, \hat{i}_{\pi(\ell)} | i, c_{\pi(1)},
\ldots, c_{\pi(\ell)}) \, ,
\]
which clearly fulfills the first condition of
Definition~\ref{def:shareable}. By assumption, we have that for all $i
\in [m]$, $P(v, \hat{i}_1, \ldots, \hat{i}_\ell | i, c_1,
\ldots, c_\ell) =0$ whenever $v \in c_k$ for all $k\in [\ell]$ and there is a $k'$ such that $i \neq \hat{i}_{k'}$. As this condition holds for each pair of Alice and Bob$_k$ individually, it is invariant under permutations of Bobs. Therefore, the same condition also holds for $Q$. 
Since any convex combination of non-signaling distributions is also non-signaling, it follows that
$Q$ can also be used to achieve $\alpha_{1,\ell}^{\ns}(G)$. We now focus
on the marginal distribution $Q(v,\hat{i}_1 | i, c_1)$ between Alice and the
first Bob. This distribution is non-signaling and $\ell$-shareable by
construction where $\ell \geq \theta'_e(G)$, \emph{i.e.}, $\ell$ is greater or equal to the number of outputs of the specific channel we consider. By Theorem~\ref{thm:ns}, $Q$ admits a
local-hidden-variable model. In other words, Alice and the first Bob
can achieve the distribution by using classical shared randomness. 
However, as we are considering the zero-error scenario, shared randomness does not improve over the deterministic classical setting.
Therefore Alice and the first Bob are unable to transmit
more than $\alpha(G)$ messages over the channel,
 showing that
$\alpha^{\ns}_{1,\ell}(G) \leq \alpha(G)$. 
The claim of the theorem then follows by combining the inequality with $\alpha(G) = \alpha_{1,\ell}(G) \leq \alpha^*_{1,\ell}(G) \leq \alpha^{\ns}_{1,\ell}(G)$.
\end{proof}

\subsection{Entanglement can improve the capacity for a finite number of receivers}

The authors of~\cite{Briet:2013} present a new lower bound technique for the entanglement-assisted Shannon capacity based on 
quantum remote state preparation \cite{Bennett:2001}.
We give a review of this technique while considering the compound channel scenario.
To simplify the presentation we will use the quantum teleportation protocol of Bennett et al.~\cite{Bennett:1993} instead of quantum remote state preparation, however the key idea of the technique remains unchanged. 
While using remote state preparation would give a slightly better lower bound, this goes beyond the precision we need for our result. 
We show that there exist graphs for which
the entanglement-assisted capacity is strictly bigger than the classical one for every fixed number of receivers: for every fixed $\ell \geq 1$ we can find a graph such that 
 $ c_{1,\ell}^*(G) > c_{1,\ell}(G)$.

\subsubsection{Lower bound by teleportation}\label{sec:tele}

We first describe the lower bound for the two-party case.
Recall that an orthogonal representation of a graph $G$ of dimension $d$ is a map $f$ from $V(G)$ to the complex $d$-dimensional unit sphere that maps adjacent vertices to orthogonal vectors.
Let $\xi(G)$ be the minimum dimension of an orthogonal representation of $G$ and $t \in \mathbb{N}$.
The lower bound is obtained using the following protocol:\begin{enumerate}
\item Alice uses the channel $t$ times to send any $t$-tuple of vertices $v_1, \dots , v_t$. Bob receives the corresponding outputs $w_1, \dots , w_t$.
\item Alice teleports the quantum states given by the orthogonal representation 
$f(v_1), \dots , f(v_t)$ 
to Bob using the shared (maximally) entangled state. In total, she teleports a state of dimension $\xi(G)^t$.
Bob now has the rotated versions of the states and needs $2t \log \xi(G)$ classical bits of information to complete the teleportation protocol.
\item\label{step3} Alice uses the channel one more time to send Bob the $2t \log \xi(G)$ classical bits of correction with zero-error (using an independent set of $G$).
\item For each $j\in [t]$, Bob uses his output $w_j$ to identify a clique of $G$ that contains $v_j$. He can then define a projective measurement on the orthogonal representations of the elements of such clique and recover $v_j$ with probability 1 by measuring $f(v_j)$.
\end{enumerate}

The protocol works if $t$ is such that $\alpha(G) \ge \xi(G)^{2t}$, implying that Step \ref{step3} can be performed. 
If this is the case,  Alice and Bob can transmit a message set of size $m=|V|^t$ using the channel $t+1$ times.

The protocol can be easily extended to the multi-receiver case.
Suppose that there are $\ell$ Bobs and that Alice shares 
an independent maximally entangled state with each of them.
With $t+\ell$ uses of the channel, Alice can perfectly communicate one out of $|V(G)|^t$ messages with zero-probability of error to all the Bobs. 
In fact, with $t$ uses of the channel, Alice sends any $t$-tuple of vertices $v_1, \dots, v_t$. She then teleports their orthogonal representations to each Bob individually and uses the channel $\ell$ additional times to send the classical bits that complete teleportation. The $k$-th Bob will consider the $(t+k)$-th use of the channel as his output and ignore the others.
%
This protocol gives the lower bound: 
\begin{equation} \label{eq:telepbound}
c^*_{1,\ell}(G) \geq \frac{1}{t+\ell} \log |V(G)|^t,
\end{equation}
where we chose $t \in \mathbb{N}$ such that $\alpha(G) \ge \xi(G)^{2t}$ (or equivalently $t\leq \frac{\log(\alpha(G))}{2\log(\xi(G))}$).

\subsubsection{Quarter orthogonality graph}\label{sec:quartorth}
In this subsection, we present a graph family for which entanglement improves the zero-error capacity with finitely many Bobs. These graphs were introduced in~\cite{Briet:2012} and, in~\cite{Briet:2013}, they were used to exhibit an infinite family of graphs for which exists separation between classical and entanglement-assisted capacity.

\begin{defn} Let $k$ be an odd positive integer. The \emph{quarter orthogonality graph} $\Gamma_k$ has as vertex set the set of vectors $u\in \{ \pm 1\}^{k+1}$ such that $u_{1} = 1$ and $u$ has an even number of -1 entries. The edge set consists of the pairs of orthogonal vectors.\footnote{$\Gamma_k$ is called quarter orthogonality graph because it is an induced subgraph of the orthogonality graph $\Omega_{k+1}$ (as defined in Section~\ref{sec:not}) with one quarter of its vertices.}
\end{defn}

For this graph, it is shown in~\cite{Briet:2013} (see \cite{Briet:2012} for a similar result) that $c(\Gamma_k) \leq 0.846 k$ holds for $k = 4p^s-1$ where $p$ is an odd prime and $s \in \mathbb{N}$. Thus, for all $\ell \in \mathbb{N}$ and $k$ as above, $ c_{1,\ell}(\Gamma_k) \leq 0.846 k $.
The map $f: V(\Gamma_k) \to \mathbb{R}^{k+1}$ such that $f(v) = v/\sqrt{k+1}$ is an orthogonal representation of $\Gamma_k$ and, hence, $\xi(\Gamma_k) \leq k+1$.
If $t$ is such that $(k+1)^{2t} \leq \alpha(\Gamma_k)$, we have $\xi(\Gamma_k)^{2t} \leq \alpha(\Gamma_k)$ and we can use the
protocol
 from the previous section to communicate $|V(\Gamma_k)|^t$ messages.
A simple lower bound on the independence number of the graph $\Gamma_k$ can be derived considering the set $U$ of vertices that have ones in their first $(k+3)/2$ coordinates. It is easy to see that $U$ is an independent set, and that its cardinality is $2^{(k-3)/2}$ by recalling that the vertices of $\Gamma_k$ have an even number of $-1$ entries, 

We show that the teleportation lower bound previously described gives a separation between $c_{1,\ell}(\Gamma_k)$ and $c^*_{1,\ell} (\Gamma_k)$ for certain $k$ and $\ell \geq 1$.

\begin{thm} For every odd integer $k\geq 5$ and $\ell \in \mathbb{N}$,
$$ c^*_{1,\ell} (\Gamma_k) \geq  \frac{t}{t+\ell} (k-1)  $$
with $t = \lfloor \frac{k-3}{4 \log(k+1)}\rfloor $.
\end{thm}

\begin{proof}
Note that if $t = \lfloor \frac{k-3}{4 \log(k+1)}\rfloor $, we have from the discussion above that $\xi(\Gamma_k)^{2t} \leq (k+1)^{2t} \leq 2^{(k-3)/2} \leq \alpha(\Gamma_k)$.
Therefore we can apply Equation~\eqref{eq:telepbound} and obtain the desired bound, since $|V(\Gamma_k)| = 2^{k-1}$.
\end{proof}

\begin{cor}\label{cor:compound}
Consider any $k = 4p^s-1$ such that $p$ is an odd prime and $s \in \mathbb{N}$.
If  
$
\ell < \frac{0.144k-1}{0.856k} \lfloor \frac{k-3}{4\log(k+1)} \rfloor
$
then $c^*_{1,\ell}(\Gamma_k) > c_{1,\ell}(\Gamma_k)$.
\end{cor}

\begin{proof}
Easy algebraic manipulations give that for every $\ell < \left(\frac{0.144k-1}{0.856k} \right) t$ where $t =  \lfloor \frac{k-3}{4\log(k+1)} \rfloor
$,
$$c^*_{1,\ell}(\Gamma_k) \geq  \frac{t}{t+\ell} (k-1) > 0.846k \geq c_{1,\ell}(\Gamma_k).$$
\end{proof}

\noindent
It follows that our lower bound, for $k\approx 1000$ is strictly larger than the classical capacity up to $\ell = 4$, for $k\approx 2000$ up to $\ell=7$, and the upper bound on $\ell$ tends to infinity as $k$ goes to infinity.

\section{Multiple senders} \label{sec:send}

In this section, we focus on a different zero-error communication scenario which can be thought as multiple senders with a single receiver.
Suppose there are $\ell$ Alices, each of whom gets access to a classical channel which connects her to the single Bob.
We are interested in the total amount of messages that the senders, as a group, can transmit perfectly. 
At every stage of the communication only one of the Alices uses her channel to communicate an input.
We assume that inputs of one sender cannot be confused with inputs from another sender. In other words, the receiver knows which one of the senders sent him the message.  
We want to study the maximum cardinality of a message set that the senders are able to perfectly communicate to the receiver when they are allowed to cooperate.
Equivalently, this communication scenario can be depicted as single-sender single-receiver where the sender can choose among $\ell$ channels $\{ \mathcal N_1, \mathcal N_2, \dots \mathcal N_\ell \}$ to use for the communication.
At every round of communication, the receiver learns both the output of the channel as well as which channel has been used.

Suppose that the $k$-th Alice is connected to Bob through a channel $\mathcal N_k$ with confusability graph $G_k$. 
As noticed in~\cite{Alon:2007}, the confusability graph related to $\ell$ cooperating Alices is given by the disjoint union $G_1 + G_2 + \dots + G_\ell$. 
(Intuitively, since inputs from different senders cannot be confused there are not edges between vertices of $G_i$ and $G_j$ if $i \neq j$.)
The maximum size of a message set that can be perfectly transmitted with one use of the channel is $\alpha (G_1 + G_2 + \dots + G_\ell) = \sum_{i=1}^\ell \alpha(G_i)$. 
What happens if we allow multiple channel uses?
Shannon~\cite{Shannon:1956} showed that for every pair of graphs $G$ and $H$, if $c(G) = \log A$ and $c(H) = \log B$ then $c(G+H)\geq \log(A+B)$ and he conjectured that equality holds.
However, Alon~\cite{Alon:1998} showed that there exists a pair of graphs for which strict inequality holds. From an information-theoretical perspective, this example says that if the two senders are allowed to cooperate, the average number of messages they can communicate is strictly more than the sum of their individual possibilities.
This result was extended by Alon and Lubetzky~\cite{Alon:2007} for a larger number of senders. They showed that it is possible to assign a channel to each sender such that only privileged subsets of senders are allowed to communicate with high capacity.

\begin{figure}[t]
\begin{center}
 \includegraphics[width=12.5 cm]{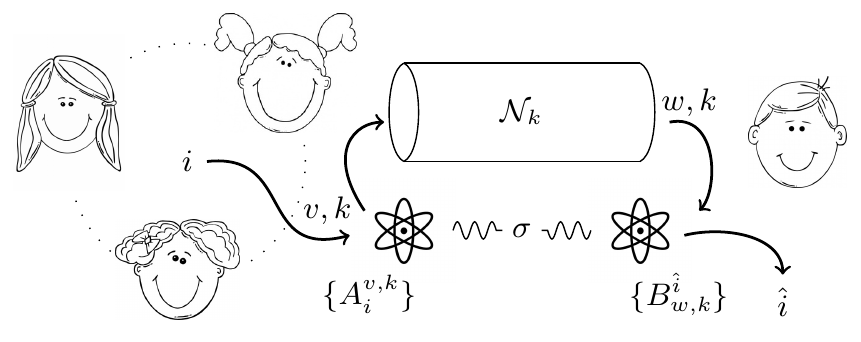}
 \caption{The figure describes an entanglement-assisted multi-sender channel instance, with a single channel use. The Alices can cooperate and the $k$-th Alice has access to a classical channel $\mathcal N_k$.
If the Alices want to communicate the message ${i}$ to Bob, one of them performs a measurement $\{A_{{i}}^{{v},k} \}$ (that depends on ${i}$) on her part of the entangled state. The outcome ${{v},k}$ indicates that the $k$-th Alice should use her channel $\mathcal N_k$ to send input ${v}$.
Bob receives an outcome ${w}$ and, by assumption, he knows that channel $\mathcal N_k$ has been used for the communication. 
He can then perform measurement $\{B_{{w},k}^{\hat{i}}\}$, which depends on ${w}$ and $k$, and outputs ${\hat{i}}$.
The protocol works if ${\hat{i}}$ is equal to ${i}$ with zero probability of error.
}
\label{fig:multi-sender}
\end{center}
\end{figure}

We study this problem in the entanglement-assisted setting focusing on the particular case where all $\ell$ Alices have access to the same channel $\mathcal N$ with confusability graph $G$. 
Since the senders are cooperating and there is no restriction on the amount of shared entanglement in this model of communication, we can assume that only one of the senders actually performs quantum operations on the entangled state.
Hence, we can without loss of generality assume the quantum state to be bipartite and the entanglement-assisted strategy to be used to send messages through a classical channel with confusability graph $G^{+ \ell}$. 
Recall that $G^{+\ell}$ denotes the disjoint union of $\ell$ copies of the graph $G$.
 The protocol is depicted in Figure~\ref{fig:multi-sender}.

\begin{defn}\label{def:qmultiuser}
For a graph $G$ define the \emph{entanglement-assisted multi-sender independence number} with $\ell$ senders as $\alpha_{\ell,1}^*(G) := \alpha^*(G^{+\ell})$.
The \emph{entangled multi-sender Shannon capacity} with $\ell$ senders is $c_{\ell,1}^*(G) := c^*(G^{+\ell})$ which by definition is equal to $ 
\lim_{n\to\infty}\frac{1}{n}\log\alpha^*((G^{+\ell})^{\boxtimes n})$.
\end{defn}
A useful observation is that $\alpha_{\ell,1}^*(G) \ge \ell \cdot \alpha^*(G) $ for every $G$ and $\ell \in \mathbb{N}$.
Indeed, each Alice can individually communicate $\alpha^*(G)$ messages using entanglement and in our model Bob learns automatically which Alice performed the communication. Therefore, the $\ell$ cooperating Alices can communicate at least one among $\ell \cdot \alpha^*(G)$ messages with one use of the channels and entanglement.
%
%
Somewhat surprisingly, we present in Section \ref{sec:viol} an example of a graph for which $\ell$ senders have a better joined strategy. In other words, there is a graph $G$ and $\ell \in \mathbb{N}$ for which $\alpha^*_{\ell,1}(G) > \ell \cdot \alpha^*(G)$.
This does not happen in the classical case where, using analogous notation,
we have $\alpha_{\ell,1}(G) := \alpha(G^{+\ell}) = \ell \cdot \alpha(G)$ and $c_{\ell,1}(G) := c(G^{+\ell}) = c(G) + \log \ell$ for every $G$ and $\ell \in \mathbb{N}$.
The latter equality is also mentioned in~\cite{Shannon:1956} and it follows from the fact that the strong graph product distributes over the disjoint union, \emph{i.e.}, $G \boxtimes (H_1+H_2) = G\boxtimes H_1 + G \boxtimes H_2$ for every $G, H_1,H_2$ (see for example~\cite{ProductGraphs} for a proof), which in particular implies that $(G^{+\ell})^{\boxtimes n} = (G^{\boxtimes n})^{+\ell^n}$.
Using this last equality, we have
\begin{align*}
c(G^{+\ell}) & =  \lim_{n \to \infty} \frac{1}{n}\log\alpha \big((G^{+\ell})^{\boxtimes n}\big) 
= \lim_{n \to \infty} \frac{1}{n}\log \alpha\big((G^{\boxtimes n})^{+\ell^n}\big) \\
& = \lim_{n \to \infty} \frac{1}{n} \log \big(\ell^n \cdot \alpha(G^{\boxtimes n})\big) 
= \lim_{n \to \infty} \frac{1}{n} \log \big(\alpha(G^{\boxtimes n})\big) + \log \ell = c(G) + \log \ell.
\end{align*}

\subsection{Separation between entanglement-assisted and classical multiple senders capacities}

In this section, we show the following: for every graph with a separation between the classical and entanglement-assisted capacity, there is also a separation in the multi-sender setting independently of the number of senders (Theorem \ref{thm:sep_c_l}). The same type of result holds when we restrict to single use of the channel (Lemma \ref{lem:sep_alpha_l}). 
The latter is to be expected since we mentioned above that there is an easy quantum strategy  that allows the $\ell$ Alices to communicate $\ell \cdot \alpha^*(G)$ messages with one use of the channels.

\begin{lem}\label{lem:sep_alpha_l}
For any graph $G$ such that $\alpha^*(G) > \alpha(G)$, we have $\alpha^*_{\ell,1}(G)>\alpha_{\ell,1}(G)$ for every $\ell \in \mathbb{N}$.
\end{lem}

\begin{proof}
Each Alice can individually communicate $\alpha^*(G)$ messages using entanglement. Since Bob also learns which Alice has sent him the message, we have an easy strategy that allows to send one among $\ell \cdot \alpha^*(G)$ messages with entanglement. 
Hence, we have 
$$\alpha^*_{\ell,1}(G) \geq  \ell \cdot \alpha^*(G) > \ell \cdot \alpha(G) = \alpha(G^{+\ell}) = \alpha_{\ell,1}(G).$$
\end{proof}


\begin{thm}\label{thm:sep_c_l}
For any graph $G$ such that $c^*(G) > c(G)$, we have $c^*_{\ell,1}(G)>c_{\ell,1}(G)$ for every $\ell \in \mathbb{N}$.
\end{thm}

\begin{proof}
Recall that $(G^{+\ell})^{\boxtimes n} = (G^{\boxtimes n})^{+\ell^n}$ and that 
$\alpha^*_{\ell,1}(G) = \alpha^*(G^{+\ell}) \geq \ell \cdot \alpha^*(G)$ for every $\ell \in \mathbb{N}$ and graph $G$. Therefore, we get
\begin{align*}
c_{\ell,1}^*(G) &= \lim_{n \to \infty} \left( \frac{1}{n} \log \left( \alpha^*\big((G^{+\ell})^{\boxtimes n}\big) \right) \right)
= \lim_{n \to \infty} \left( \frac{1}{n} \log \left( \alpha^*\big((G^{\boxtimes n})^{+\ell^n}\big) \right) \right) \\
& \ge \lim_{n \to \infty} \left( \frac{1}{n} \log \left(\ell^n \cdot \alpha^*(G^{\boxtimes n})\right) \right) 
= \lim_{n \to \infty} \left( \frac{1}{n} \log (\alpha^*(G^{\boxtimes n})) + \log \ell \right) \\
& = \lim_{n \to \infty} \left(\frac{1}{n} \log \left(\alpha^*(G^{\boxtimes n})\right) \right) + \log\ell
 = c^*(G) + \log\ell\\
&> c(G) + \log\ell 
= c_{\ell,1}(G).
\end{align*}
%
\end{proof}

\subsection{Improving communication by joint entanglement-assisted strategy}\label{sec:viol}

In this section we show the existence of a graph $G$ and natural number $\ell$ for which $\alpha^*_{\ell,1}(G) > \ell \cdot \alpha^*(G)$. This means that there is an entanglement-assisted strategy for $\ell$ cooperating senders which is strictly better than the sum of their best individual strategies.
More generally, we are able to prove that there exist graphs for which cooperation among the senders allows a better entanglement-assisted strategy for any finite number of channels uses. 
This is a peculiar property of the entanglement-assisted setting since in the classical case $\alpha_{\ell,1}(G)= \ell \cdot \alpha(G)$ and $c_{\ell,1}(G) = c(G) + \log\ell$ always hold.
We currently do not know whether this improvement gained by cooperation in the entanglement-assisted setting extends also to the asymptotic regime.



In order to prove the result, we need to briefly describe a different two-party entanglement-assisted communication scenario. 
We defer to \cite{Briet:2013} for a more detailed explanation.

Let $G$ be a graph. Suppose that Alice receives a vertex $x \in V(G)$ and Bob receives (as side information) a clique $\mathcal{C} \subseteq V(G)$ under the promise that $x \in \mathcal{C}$. Alice can send classical messages to Bob without error.
What is the minimum cardinality of a message set that Alice has to use to communicate to Bob such that he can perfectly learn Alice's input $x$?
It is straightforward to check that in the classical scenario, the minimum cardinality is given by the chromatic number $\chi(G)$. 
Given an optimal coloring of the graph, Alice can simply send the color corresponding to $x$ to Bob. By definition, the color is sufficient to learn $x$, because elements of a clique all have different colors. Conversely, any deterministic strategy yields a coloring of the graph. We can assume an optimal strategy to be deterministic as we are considering the zero-error scenario.  
Similarly, the \emph{entanglement-assisted chromatic number} $\qchrom(G)$ is the minimum cardinality of a message set that Alice has to send to Bob such that he can perfectly learn $x$ when Alice and Bob can share an arbitrary entangled state.
The parameter $\qchrom(G)$ was introduced in~\cite{Briet:2013} to quantify the amount of communication needed in the above mentioned scenario. 
There it is shown that $\vartheta(\overline{G}) \leq \qchrom(G)$ and  $\qchrom(G^{\boxtimes m})\leq \qchrom(G)^m$ for every graph $G$ and $m\in \mathbb{N}$. 

We will need the following technical lemma. Recall that $G \square K_t$ is the Cartesian product graph between $G$ and the complete graph $K_t$.

\begin{lem}\label{lem:qchromqindep}
For any graph $G$, if $t = \qchrom(G)$ then $\qindep(G \square K_t) = |V(G)|.$
\end{lem}

\begin{proof}
First, we prove the inequality $\qindep(G \square K_t) \leq |V(G)|$ for any $t \in \mathbb{N}$.
Let $|V(G)| = n$. For every $t \in \mathbb{N}$, we have $\qindep(G \square K_t) \leq \vartheta(G \square K_t) \leq \vartheta(\overline{K}_{n} \boxtimes K_t)=\vartheta(\overline{K}_{n}) \cdot \vartheta(K_t) = n = |V(G)|$. This chain of inequalities uses the fact that $\vartheta$ upper bounds $\alpha^*$, $\overline{K}_{n} \boxtimes K_t$ is a subgraph of $G \square K_t$ and $\vartheta$ is monotone non-decreasing under taking subgraphs, $\vartheta$ is multiplicative under strong graph products and that $\vartheta(K_t) = 1$ and $\vartheta(\overline{K}_n) = n$.

Let $t = \qchrom(G)$ and suppose that Alice and Bob are connected through a one-way classical channel with confusability graph $G \square K_t$.
We present a strategy which allows Alice to communicate $|V(G)|$ messages through this channel with the help of entanglement, thus implying that $\qindep(G \square K_t) \geq |V(G)|$.
Suppose that Alice wants to send message $x \in V(G)$ to Bob. 
Using the strategy for the entanglement-assisted chromatic number $\qchrom(G)$, Alice makes a measurement on her part of the entangled state and gets an outcome $i \in [\qchrom(G)]$.
She sends message $(x,i)$ through the channel. 
To any channel output $w \in W$ we can associate a clique in the confusability graph, describing the set of messages that are confusable to Bob given $w$. There are two types of cliques in $G \square K_t$: either $\{ (z,h): z$  is in a clique $ \mathcal{C}$ of $G\}$ or $\{ (z,h) : h \in H \subseteq  [\qchrom(G)] \}$.
Suppose that from his channel output Bob infers that Alice's input $(x,i)$ is an element of the set $\{(y,i) : y \in \mathcal{C}_x $ where $\mathcal{C}_x$ is a clique of $G$ containing $x\}$.
Since Bob learns $\mathcal{C}_x$, he can use message $i$ to finish the protocol of the entanglement-assisted chromatic number. As mentioned above the protocol allows Bob to recover $x$ with zero probability of error.
For the other case, Bob from his output learns that Alice's input is an element in the set $\{ (x,j) : j \in J \subseteq [\qchrom(G)]$ with $i \in J \}$. Then he can directly deduce that $x$ is the original message used by Alice.
Hence, we have shown an entanglement-assisted protocol that allows to perfectly communicate $|V(G)|$ classical messages through a channel with confusability graph $G \square K_t$.
We conclude that if $t = \qchrom(G)$ then $\alpha^*(G \square K_t) \ge |V(G)|$.
Combining this last inequality with the one derived at the beginning of the proof, we can conclude.
\end{proof}

As a side remark, there is an analogous classical counterpart of the statement in Lemma \ref{lem:qchromqindep}: if $t = \chi(G)$ then $\alpha(G \square K_t) = |V(G)|$. Its proof can be derived using a reasoning similar to the proof above or by a simple graph-theoretic argument.

The following lemma is proven in a more general context in~\cite{Gvozdenovic:2008}.

\begin{lem}\label{lem:sepdisjoint}
Suppose $G$ is a graph such that $\qchrom(G) \cdot \qindep(G) < |V(G)|$ and let $t = \qchrom(G)$. Then, $\qindep(G^{+t}) > t \cdot \qindep(G)$.
\end{lem}

\begin{proof}
From Lemma~\ref{lem:qchromqindep}, we get $ t \cdot \qindep(G)  = \qchrom(G) \cdot \qindep(G)< |V(G)|=\qindep(G \square K_t) \leq \qindep(G^{+t})$ where the last inequality follows from the fact that $G^{+t}$ is a subgraph of $G \square K_t$ and that $\qindep$ is monotone non-decreasing under taking subgraphs.
\end{proof}

Recall from Section~\ref{sec:not} that the orthogonality graph $\Omega_k$ has all the vectors $\{ \pm 1\}^k$ as vertex set and two vectors are adjacent if orthogonal. 
From~\cite{Roberson:2012}, we know that $\vartheta(\Omega_k) = 2^k/k$ and $\vartheta(\overline{\Omega_k}) = k$  if $k$ is a multiple of four. 
Consider the orthogonal representation $f:V(\Omega_k)\to \mathbb{R}^k$ with $f(v) = v/\sqrt{k}$ that maps vertices of $\Omega_k$ to the unit sphere and adjacent vertices to orthogonal vectors.
Since $\qchrom$ is upper bounded by the minimum dimension of an orthogonal representation in which all the entries of the vectors have equal moduli~\cite{Briet:2013}, we have that $k = \vartheta(\overline{\Omega_k}) \leq \qchrom(\Omega_k) \leq k$, and thus $\qchrom(\Omega_k) = k$, for every $k$ multiple of four. 

We show the existence of a graph $G$ and $\ell \in \mathbb{N}$ for which $\qindep_{\ell,1}(G) > \ell \cdot \qindep(G)$. Hence, $\ell$ cooperating senders can communicate strictly more (with one use of a channel and entanglement) than the sum of what they can communicate individually.

\begin{thm}\label{thm:disjointalpha}
Let $\Omega_k$ be the orthogonality graph with $k$ a multiple of four but not a power of two. Then $ \alpha^*_{k,1}(\Omega_k) = \alpha^*(\Omega_k^{+k}) >k \cdot \alpha^*(\Omega_k)$.
\end{thm}

\begin{proof}
From the reasoning above we know that $\qchrom(\Omega_k) = k$.
Moreover, $\qindep(\Omega_k) \leq \floor{\vartheta(\Omega_k)} = \floor{2^k/k} < 2^k/k = \vartheta (\Omega_k)$.
Using a similar argument as in~\cite{Roberson:2012}, we get that  $\qchrom(\Omega_k) \cdot \qindep(\Omega_k) < |V(\Omega_k)|$ since 
$$\qchrom(\Omega_k) \cdot \qindep(\Omega_k) \leq k \cdot \floor{2^k/k} < k \cdot 2^k/k  = 2^k = |V(\Omega_k)|.$$
Using Lemma~\ref{lem:sepdisjoint} we conclude that $\alpha^*_{k,1}(\Omega_k) = \alpha^*(\Omega_k^{+k})> k \cdot \alpha^*(\Omega_k)$.
\end{proof}

%
%

With a similar reasoning, we can prove that for every finite number of uses of the channel, cooperation among the players improves the entanglement-assisted communication.
Let $\alpha^*_{\ell,1}(G,n) := \alpha^*((G^{+\ell})^{\boxtimes n})$ be the maximum cardinality of a message set that $\ell$ Alices can communicate perfectly to Bob with $n$ uses of the channel and entanglement. 

In the next lemma, we show that there exist a graph $G$ and $\ell \in \mathbb{N}$ such that $\alpha^*_{\ell,1}(G,n) > \ell^n \cdot \alpha^*(G^{\boxtimes n})$ for every $n \in \mathbb{N}$.
This is equivalent to saying that there exists a channel and a certain number of senders for which cooperation among the senders strictly improves the communication of $n$ channel uses for every $n \in \mathbb{N}$. 

\begin{thm}\label{thm:disjointalpha_n}
Let $\Omega_k$ be the orthogonality graph with $k$ a multiple of four but not a power of two. Then, $\alpha_{k,1}^*(\Omega_k,n) > k^n \cdot \alpha^*(\Omega_k,n)$ for every $n \in \mathbb{N}$.
\end{thm}

\begin{proof}
Using the properties of Lov\'asz theta number presented in Section \ref{sec:not}, we have that $\vartheta (\Omega_k^{\boxtimes n}) = \vartheta (\Omega_k)^n = \left(\frac{2^k}{k}\right)^n$ and $\vartheta (\overline{\Omega_k^{\boxtimes n}}) = \vartheta (\overline{\Omega_k})^n = k^n$ for every $n \in \mathbb{N}$.
Then by sub-multiplicativity of $\qchrom(G)$~\cite{Briet:2013} and since $\chi^*(\Omega_k) = k$, we have $k^n = \vartheta (\overline{\Omega_k^{\boxtimes n}}) \leq \chi^*(\Omega_k^{\boxtimes n}) \leq \chi^*(\Omega_k)^n = k^n$.
This implies that for any integer $n$, 
$$\alpha^*(\Omega_k^{\boxtimes n}) \leq \floor{\vartheta (\Omega_k^{\boxtimes n})} = \Big\lfloor\left(\frac{2^k}{k}\right)^n\Big\rfloor < \left(\frac{2^k}{k}\right)^n = \frac{|V(\Omega_k^{\boxtimes n})|}{\chi^*(\Omega_k^{\boxtimes n})}.$$ 
Applying Lemma~\ref{lem:sepdisjoint}, we then have that $\alpha^*\big((\Omega_k^{\boxtimes n})^{+k^n}\big) > k^n \cdot \alpha^*(\Omega_k^{\boxtimes n})$ for every $n \in \mathbb{N}$.
We can conclude that 
$$\alpha_{k,1}^*(\Omega_k,n) = \alpha^*\big((\Omega_k^{\boxtimes n})^{+k^n}\big) >k^n \cdot \alpha^*(\Omega_k^{\boxtimes n}) = k^n \cdot \alpha^*(\Omega_k,n).$$
\end{proof}
%



\section{Conclusions}\label{sec:con}


We have studied the effects of entanglement in two multi-party channel-coding problems. 
For the compound channel setting (with one sender and multiple receivers) we have shown that entanglement can only help if the number of receivers is below a certain threshold which depends only on the channel. If there are more receivers, entanglement does not help for zero-error communication for any finite number of channel uses.
In the second situation, where there are multiple senders and one receiver, we have shown that there are channels for which entanglement always improves the communication.

In both these situations, we assume that the multiple parties have access to identical channels. 
The first natural question to ask is about the effect of entanglement when the channels are different. 
Is it possible to extend to the entanglement-assisted setting the results obtained by \cite{Gargano:1994} (for multiple receivers) and \cite{Alon:2007} (for multiple cooperating senders)?
These problems seem difficult as we only have a limited understanding of the behavior of the parameters $\alpha^*$ and $c^*$.

Two more specific questions are related to Theorem \ref{thm:qmonogamy} in Section \ref{sec:rec-neg}. 
Can we find a better bound on the number of receivers for which the entanglement-assisted capacity is equal to the classical capacity in the compound channel scenario?
The bound we have obtained is the edge-clique cover number $\theta'_e$, but this bound is derived using non-signaling distributions and therefore in a more general context than the entanglement-assisted setting.
Furthermore, it would be interesting to know whether an asymptotic version of Theorem \ref{thm:qmonogamy} holds. Does the parameter $c_{1,\ell}^*(G)$ tend to the classical capacity $c(G)$ when the number of receivers $\ell$ goes to infinity?

In general, finding better bounds for the parameter $c^*$ is interesting. 
Can we find a new general protocol (like the one based on teleportation in Section \ref{sec:tele}) that gives a better (or incomparable) lower bound on $c^*$?
In a similar spirit, can we find an upper bound on $c^*$ which is different from the Lov\'asz theta number? 
Note that such an upper bound is known for the classical parameter $c$ and it was found by Haemers \cite{Haemers:1978}.
A related question is to find a graph $G$ such that $c^*(G) < \log \vartheta(G)$.
An approach to this latter problem is to find a pair of graphs, $G$ and $H$ (not necessarily distinct), for which $c^*(G+H) > \log (A + B)$  where $c^*(G) = \log A$ and $c^*(H) = \log B$. Such a result would be in the same spirit as our findings on the parameter $\alpha^*$ described in Section \ref{sec:viol}.


\paragraph{Acknowledgments}

The authors thank Jop Bri\"et, Harry Buhrman, Monique Laurent, Laura Man\v{c}inska, Fernando de Melo, Andreas Winter and Ronald de Wolf for useful discussions. We also thank the anonymous referees for their excellent comments and suggestions to improve this article, in particular for the operational proof of Lemma~\ref{lem:qchromqindep}. 
T.P. was partially funded by the European project SIQS.  G.S. was supported by the European Commission (STREP RAQUEL). Part of this work was done when G.S. was a PhD student at CWI Amsterdam, supported by Ronald de Wolf's VIDI grant from NWO. C.S. was supported by an NWO VENI grant.


\begin{thebibliography}{50}

\bibitem[AL06]{Alon:2006}
N.~Alon and E.~Lubetzky.
\newblock The Shannon capacity of a graph and the independence numbers of its
  powers.
\newblock {\em IEEE Transactions on Information Theory}, 52(5):2172--2176,
  2006.

\bibitem[AL07]{Alon:2007}
N.~Alon and E.~Lubetzky.
\newblock Privileged users in zero-error transmission over a noisy channel.
\newblock {\em Combinatorica}, 27(6):737--743, 2007.

\bibitem[Alo98]{Alon:1998}
N.~Alon.
\newblock The Shannon capacity of a union.
\newblock {\em Combinatorica}, 18:301--310, 1998.

\bibitem[BBC{\etalchar{+}}93]{Bennett:1993}
C.~H. Bennett, G.~Brassard, C.~Cr\'{e}peau, R.~Jozsa, A.~Peres, and W.~K.
  Wootters.
\newblock Teleporting an unknown quantum state via dual classical and
  \uppercase{E}instein-\uppercase{P}odolsky-\uppercase{R}osen channels.
\newblock {\em Physical Review Letters}, 70:1895--1899, 1993.

\bibitem[BDS{\etalchar{+}}01]{Bennett:2001}
C.~H.~Bennett, D.~P.~DiVincenzo, P.~W.~Shor, J.~A.~Smolin, B.~M.~Terhal, and W.~K.~Wootters.
\newblock {Remote State Preparation}.
\newblock {\em Physical Review Letters}, 87:077902, 2001.

\bibitem[BSST02]{Bennett:2002}
C.~H. Bennett, P.~W.~Shor,  J.~A.~Smolin, and A.~V. Thapliyal.
\newblock Entanglement-assisted capacity of a quantum channel and the reverse Shannon theorem.
\newblock {\em IEEE Transactions on Information Theory}, 48:2637--2655, 2002.


\bibitem[BBG12]{Briet:2012}
J.~Bri\"{e}t, H.~Buhrman, and D.~Gijswijt.
\newblock {Violating the Shannon capacity of metric graphs with entanglement}.
\newblock {\em Proceedings of the National Academy of Sciences}, 2012.

\bibitem[BBL{\etalchar{+}}13]{Briet:2013}
J.~Bri\"et, H.~Buhrman, M.~Laurent, T.~Piovesan, and G.~Scarpa.
\newblock {Zero-error source-channel coding with entanglement}.
\newblock{ In {\em The Seventh European Conference on Combinatorics, Graph Theory and Applications},
pages 157-162, 2013.} 
\newblock  Full version available at arXiv:1308:4283.


\bibitem[Bei10]{ZEC-theta}
S.~Beigi.
\newblock Entanglement-assisted zero-error capacity is upper-bounded by the
  {L}ov\'{a}sz theta function.
\newblock {\em Physical Review A}, 82:10303--10306, 2010.

\bibitem[CKS90]{Cohen:1990}
G.~Cohen, J.~K\"orner, and G.~Symonyi.
\newblock Zero-error capacities and very different sequences.
\newblock {\em Sequences: combinatorics,
  compression, security and transmission}, 144--155, Springer-Verlag,
  1990.

\bibitem[CLMW10]{CLMW09}
T.~S. Cubitt, D.~Leung, W.~Matthews, and A.~Winter.
\newblock Improving zero-error classical communication with entanglement.
\newblock {\em Physical Review Letters}, 104:230503--230506, 2010.

\bibitem[CLMW11]{CLMW11}
T.~S. Cubitt, D.~Leung, W.~Matthews, and A.~Winter.
\newblock {Zero-error channel capacity and simulation assisted by non-local correlations.}
\newblock {\em IEEE Transactions on Information Theory}, 57(8):5509 --5523, 2011.


\bibitem[DSW13]{Duan:2013}
R.~Duan, S.~Severini, and A.~Winter.
\newblock Zero-error communication via quantum channels, noncommutative graphs,
  and a quantum {Lov\'asz} number.
\newblock {\em IEEE Transactions on Information Theory}, 59(2):1164 --1174, 2013.

\bibitem[GKV94]{Gargano:1994}
L.~Gargano, J.~K{\"o}rner, and U.~Vaccaro.
\newblock Capacities: From information theory to extremal set theory.
\newblock {\em Journal of Combinatorial Theory, Series A}, 68(2):296--316, 1994.

\bibitem[GL08]{Gvozdenovic:2008}
N.~Gvozdenovic and M.~Laurent.
\newblock The operator {$\Psi$} for the chromatic number of a graph.
\newblock {\em SIAM Journal on Optimization}, 19(2):572--591, July 2008.

\bibitem[Hae78]{Haemers:1978}
W. H. Haemers.
\newblock {An upper bound for the Shannon capacity of a graph}.
\newblock {\em Colloquia Mathematical Societatis J\'anos Bolyai}, 25:267--272, 1978.

\bibitem[HIK11]{ProductGraphs}
R.~Hammack, W.~Imrich, and S.~Klav\u{z}ar.
\newblock {\em Handbook of product graphs}.
\newblock Discrete Mathematics and its Applications. CRC Press,
  Boca Raton, FL, second edition, 2011.

\bibitem[KD93]{Knuth:1993}
D. E. Knuth.
\newblock {The sandwich theorem}.
\newblock {\em The Electronic Journal of Combinatorics}, 1:1-48, 1994.

\bibitem[KO98]{Korner}
J.~Korner and A.~Orlitsky.
\newblock Zero-error information theory.
\newblock {\em IEEE Transactions on Information Theory}, 44(6):2207--2229, 1998.

\bibitem[LMM{\etalchar{+}}12]{LMMOR}
D.~Leung, L.~Man\v{c}inska, W.~Matthews, M.~Ozols, and A.~Roy.
\newblock Entanglement can increase asymptotic rates of zero-error classical
  communication over classical channels.
\newblock {\em Communications in Mathematical Physics}, 311:97--111, 2012.

\bibitem[Lov79]{Lovasz}
L.~Lov{\'a}sz.
\newblock On the {S}hannon capacity of a graph.
\newblock {\em IEEE Transactions on Information Theory}, 25(1):1--7, 1979.

\bibitem[MAG06]{MAG06}
Ll. Masanes, A.~Acin, and N.~Gisin.
\newblock General properties of nonsignaling theories.
\newblock {\em Physical Review A}, 73:012112, Jan 2006.

\bibitem[RM12]{Roberson:2012}
D. E. {Roberson} and L.~{Man\v{c}inska}.
\newblock {Graph Homomorphisms for Quantum Players}.
\newblock { arXiv:1212.1724}, 2012.

\bibitem[Sha56]{Shannon:1956}
C.~E. Shannon.
\newblock The zero error capacity of a noisy channel.
\newblock {\em IRE Transactions on Information Theory}, IT-2(3):8--19,
  1956.

\bibitem[Sin09]{Sinaimeri}
B.~Sinaimeri.
\newblock {\em Structures of diversity}.
\newblock PhD thesis, Sapienza University, Rome, 2009.


\end{thebibliography}

\newcommand{\etalchar}[1]{$^{#1}$}

\end{document}